\documentclass[12pt]{article}
\usepackage{amsmath,amsfonts,amsthm,amssymb,color}
\usepackage{fancybox}

\newtheorem{theorem}{Theorem}[section]
\newtheorem{corollary}[theorem]{Corollary}
\newtheorem{lemma}[theorem]{Lemma}

\theoremstyle{definition}
\newtheorem{definition}[theorem]{Definition}
\newtheorem{remark}[theorem]{Remark}

\newenvironment{fminipage}%
  {\begin{Sbox}\begin{minipage}}%
  {\end{minipage}\end{Sbox}\fbox{\TheSbox}}

\newcommand\vv{\overrightarrow{v}}
\newcommand\chiv{\overrightarrow{\chi}}

\title{Constant Arboricity Spectral Sparsifiers}
\author{
Timothy Chu
\\Google \\timchu@google.com
\and
Michael B. Cohen
\\M.I.T \\micohen@mit.edu
\and
Jakub W. Pachocki
\\Carnegie Mellon University
\\pachocki@cs.cmu.edu
\and
Richard Peng
\\M.I.T
\\rpeng@mit.edu
}
\begin{document}
\maketitle
\begin{abstract}
We show that every graph is spectrally similar to the union of a constant number of forests. Moreover, we show that Spielman-Srivastava sparsifiers are the union of $O(\log n)$ forests. This result can be used to estimate boundaries of small subsets of vertices in nearly optimal query time.
\end{abstract}
\section{Introduction}

A sparsifier of a graph G = (V, E, $c$), is a sparse graph $H$ with simliar properties. In this paper, we consider \textbf{spectral sparsifiers}, defined by Spielman and Teng in \cite{SpielmanT04}. A graph $H$ is said to be a $(1+\epsilon)$-spectral sparsifier of a graph $G$ if, for all vectors $x \in \mathbb{R}^{|V|}$,

\begin{equation}\label{1}
\left(\frac{1}{1+\epsilon}\right) x^TL_Gx \leq x^TL_Hx \leq (1+\epsilon)x^TL_Gx.
\end{equation}

Here, $L_G$ and $L_H$ are the Laplacians of $G$ and $H$ respectively.

\vspace{2 mm}

Spielman and Srivastava proved in \cite{SpielmanS08} that every graph has a spectral sparsifier with $O\left(\frac{n \log n}{\epsilon^2}\right)$ edges using an edge sampling routine. Batson, Spielman, and Srivastava proved in \cite{BatsonSS09} that there exist $(1+\epsilon)$-spectral sparsifiers of graphs with $O(n/\epsilon^2)$ edges.  Furthermore, the Marcus-Spielman-Srivastava proof of the Kadison-Singer conjecture in \cite{MarcusSS13} can be used to show that the edge sampling routine of \cite{SpielmanS08} gives an $(1+\epsilon)$-spectral sparsifier with $O(n/\epsilon^2)$ edges, with non-zero probability. \cite{Srivastava13}

Our primary result is to show that Spielman-Srivastava sparsifiers can be written as the union of $O\left(\frac{\log n}{\epsilon^2}\right)$ forests, and that Marcus-Spielman-Srivastava sparsifiers can be written as the union of $O(1/\epsilon^2)$ forests. This is as tight a bound as we can hope for up to $\epsilon$ factors.

Our result can be applied to approximating cut queries efficiently.
As shown by Andoni, Krauthgamer, and Woodruff in \cite{AndoniKW14}, any sketch of a graph that w.h.p. preserves all cuts in an $n$-vertex graph must be of size $\Omega(n/\epsilon^2)$ bits.
We show that the Spielman-Srivastava sparsifiers, in addition to achieving nearly optimal construction time and storage space, can also be made to achieve the nearly optimal query time $O\left(|S|\frac{\log n}{\epsilon^2}\right)$ when estimating the boundary of $S \subseteq V$, compared to the trivial query time of $O\left(n\frac{\log n}{\epsilon^2}\right)$.
\section{Preliminaries}
\subsection{Electrical Flows and Effective Resistance}

Let graph G = (V, E, $c$) have edge weights $c_e$, where $c_e$ is the conductance of each edge.  Define the resistance $r_e$ on each edge to be $\frac{1}{c_e}$. 

Let $L_G$ be the Laplacian of graph $G$. Let $\vv \in \mathbb{R}^|V|$ be the vector of voltages on the vertices of $V$.

Let the vector $\chiv$ denote the vector of excess demand on each vertex. It's well known in the theory of electrical networks that
\begin{equation}
L_G \vv = \chiv,
\end{equation}

or equivalently,
\begin{equation}
\vv = L_G^{+} \chiv
\end{equation}

where $L_G^{+}$ is the Moore-Penrose pseudo-inverse of $L_G$.

For edge $e = (i, j)$ with $i, j \in V$, the effective resistance $R_e(G)$ is defined as

\begin{equation}
R_e(G) = \frac{\chiv^TL_G^{+} \chiv}{2}
\end{equation}

where 
\begin{equation}\label{chiv}
   \chiv := \left\{
     \begin{array}{ll}
       1  &  x = i\\
       -1 &  x = j\\
       0  & \text{otherwise} 
     \end{array}
   \right.
\end{equation} 

The effective resistance of edge $e$ can be interpreted as the voltage drop across that edge given an flow of $1$ unit of current from $i$ to $j$. 

When the underlying choice of graph $G$ is clear, $R_e(G)$ will be shortened to $R_e$.

\begin{lemma} \label{effectresist}
For all graphs $H$ that spectrally sparsify $G$, 
\begin{equation}
\left(\frac{1}{1+\epsilon}\right) R_e(H) < R_e(G) < (1+\epsilon) R_e(H)
\end{equation}

\begin{proof} This follows immediately from Equation \ref{1}, and substituting
$$R_e(G) = \frac{\chiv^TL_G^{+} \chiv}{2}$$
and
$$R_e(H) = \frac{\chiv^TL_H^{+} \chiv}{2}.$$
\end{proof}
where $R_e(H)$ denotes the effective resistance of edge $e$ in $H$ and $R_e(G)$ denotes the effective resistance of $e$ in $G$.
\end{lemma}
\subsection{The Spielman-Srivastava Sparsifier}

Spielman and Srivastava showed in \cite{SpielmanS08} that any graph can be sparsified with high probability using the following routine, for a large enough constant $C$: 

\begin{itemize}
\item For each edge, assign it a probability $p_e := \frac{R_ec_e}{(n-1)}$, where $R_e$ is the effective resistance of that edge and $c_e$ is the conductance (the inverse of the actual resistance) of that edge. Create a distribution on edges where each edge occurs with probability equal to $p_e$.
\item Weight each edge to have conductance $\frac{c_e\epsilon^2}{(C n \log n) p_e}$, and sample $Cn \log n/\epsilon^2$ edges from this distribution.
\end{itemize}

We call such a scheme a Spielman-Srivastava sparsifying routine. Note that this scheme allows for multiple edges between any two vertices.
\begin{remark} Sampling by \textit{approximate} effective resistances (as Spielman and Srivastava did in their original paper \cite{SpielmanS08}) will work in place of using exact values for effective resistances. The results in our paper will still go through; an approximation will still ensure that every edge has a relatively large weighting, which is what the result in our paper depends on.
\end{remark}

\subsection{The Marcus-Spielman-Srivastava Sparsifier}
The following scheme from \cite{Srivastava13} produces a sparsifier with non-zero probability, for sufficiently large constants $C$:

\begin{itemize}
\item For each edge, assign it a probability $p_e := \frac{R_ec_e}{(n-1)}$, where $R_e$ is the effective resistance of that edge and $c_e$ is the conductance (inverse of actual resistance) of that edge. Create a distribution on edges where each edge occurs with probability equal to $p_e$.
\item Weight each edge to have conductance $\frac{c_e\epsilon^2}{p_e(Cn)}$, and sample $C n/\epsilon^2$ edges from this distribution.
\end{itemize}

We call such a scheme a Marcus-Spielman-Srivastava sparsifying routine. Note that this scheme allows for multiple edges between any two vertices. 

Note that the probability this routine returns a sparsifier may be expontentially small, and there is no known efficient algorithm to actually find such a sparsifier, making the \cite{SpielmanS08} result more algorithmically relevant.
\subsection{Uniform Sparsity and Low Arboricity}

\begin{definition}
The \textbf{arboricity} of a graph $G$ is the equal to the minimum number of forests its edges can be decomposed into.
\end{definition}
\begin{definition}
A graph G = (V, E, $c$) is said to be \textbf{$c$-uniformly sparse} if, for all subsets $V' \subset V$, the subgraph induced on $G$ by $V'$ contains no more than $c \cdot |V'|$ edges.
\end{definition}

\begin{lemma} \label{lowarb}
Uniform Sparsity implies Low Arboricity. That is, if $G$ is $c$-uniformly sparse, then the arboricity of $G$ is no greater than $2c$.
\end{lemma}
This statement is proven in \ref{lem_treelike}.

\section{The Main Result}

First we establish some preliminary lemmas.
\begin{lemma} \label{foster}
\textit{(Foster's Resistance Theorem)} Let $G = (V, E, c)$ be any graph. Then 
\begin{equation}
\sum_{e \in E} R_ec_e = n-1.
\end{equation}
where $n:=|V|$. \text{\cite{Foster61}}
\end{lemma}

\begin{lemma} \label{monotone}
\textit{(Effective Resistances of edges in a subgraph are higher than in the original graph)} Let $H$ be a subgraph of $G = (V, E, c)$, where $L_H$ is treated as a linear operator from $\mathbb{R}^{|V|}$ to $\mathbb{R}^{|V|}$. Then

\begin{equation}\label{lem} 
x^T L_H^+ x \geq x^T L_G^+ x
\end{equation}

for all $x \in \mathbb{R}^{|V|}$ orthogonal to the nullspace of $L_H$. 

\end{lemma}

\begin{proof}
Let $x = L_G^{\frac{1}{2}}y$ and $y=(L_G^+)^\frac{1}{2}x$. Now Equation \ref{lem} is equivalent to the equation

\begin{equation}
y^T \left(L_G^\frac{1}{2} L_H^+ L_G^\frac{1}{2}\right) y \geq y^T y
\end{equation}

holding true for all vectors $y \in \mathbb{R}^{|V|}$.

Since $x$ is assumed to be orthogonal to the nullspace of $L_H$ (which equals the nullspace of $L_H^+$), it follows that $y$ is orthogonal to the nullspace of $L_G^\frac{1}{2} L_H^+ L_G^\frac{1}{2}$. Therefore it suffices to show that all the non-zero eigenvalues of $L_G^\frac{1}{2} L_H^+ L_G^\frac{1}{2}$ are greater than $1$. 
This is equivalent to showing that all the non-zero eigenvalues of $\left(L_G^+ \right)^\frac{1}{2} L_H \left(L_G^+ \right)^\frac{1}{2}$ are less than $1$, a fact which follows immediately from Rayleigh monotonicity.
\end{proof}

\begin{lemma} \label{sumresist}
Let $V' \subset V$. Let $G'$ be the subgraph of $G$ induced by $V'$, and let $E'$ be the edges of the induced subgraph. 
Then 
\begin{equation}
\sum_{e \in E'} R_e(G)c_e \leq |V'|-1
\end{equation}
\end{lemma}

\begin{proof} Note that $R_e = \frac{1}{2} \chiv L_G \chiv$, where $\chiv$ is defined as in Equation \ref{chiv}. Since $e$ is an edge of subgraph $H$, it follows that $\chiv$ is orthogonal to the nullspace of $L_H$. Thus we can apply Lemma \ref{monotone} and Lemma \ref{foster} to show that

\begin{equation}
\sum_{e \in E'} R_e(G)c_e \leq \sum_{e \in E'} R_e(G')c_e =|V'|-1
\end{equation}

as desired.
\end{proof}

\begin{theorem}\label{main}
Marcus-Spielman-Srivastava sparsifiers are $O(1/\epsilon^2)$-uniformly sparse.
\end{theorem}

\begin{proof} 
For each edge included in the graph by the Marcus-Spielman-Srivastava sampling scheme, they're included in the graph with weight $\frac{c_e\epsilon^2}{p_e(Cn)}$.  Therefore, by Lemma \ref{effectresist}, the value of $R_e(H)c_e(H)$ on edge $e$ is within a $(1+\epsilon)$ multiple of

\begin{equation}
R_e(G)c_e(H) = R_e \cdot \frac{c_e \epsilon^2}{(Cn)p_e} = R_e \cdot \frac{c_e\epsilon^2}{Cn} \cdot\frac{(n-1)}{R_ec_e} \geq \epsilon^2/(2C).
\end{equation}

Here, $R_e(H)$ and $c_e(H)$ denote the effective resistance and conductance of edge $e$ in graph $H$ respectively, and $R_e$ and $c_e$ denote the effective resistance and conducatnce of edge $e$ in graph $G$ respectively.

\vspace{2 mm}
Using Lemma \ref{sumresist}, it follows that the subgraph induced by $V'$ has no more than $2C(|V'|-1)/\epsilon^2$ edges. This implies that any subgraph of a Marcus-Spielman-Srivastava sparsifier is sparse.
\end{proof}
\begin{corollary}
Marcus-Spielman-Srivastava sparsifiers have $O(1/\epsilon^2)$ arboricity.
\end{corollary}

\begin{theorem} 
Spielman-Srivastava sparsifiers are $O\left(\frac{\log n}{\epsilon^2}\right)$-uniformly sparse.
\end{theorem}
\begin{proof}
The proof is identical to the the proof of Theorem \ref{main}, with $Cn$ replaced with $Cn\log n$.
\end{proof}
\begin{corollary}
Spielman-Srivastava sparsifiers have arboricity $O(\frac{\log n}{\epsilon^2})$.
\end{corollary}

\section{Applications to Approximating Cut Queries}

\begin{definition}
	We say that a total ordering $\dot{<}$ of the vertices of a graph is \emph{$c$-treelike} if every vertex $u$ has at most $c$ neighbors $v$ such that $u~\dot{<}~v$. 
\end{definition}
\begin{lemma}
	\label{lem_treelike}
	Every $c$-uniformly sparse graph has a $2c$-treelike ordering. Moreover, this ordering can be computed in linear time.
\end{lemma}
\begin{proof}
	Let $G$ be a $c$-uniformly sparse graph.
	Let $v$ be the minimum degree vertex of $G$.
	Note that $d(v) \leq 2c$.
	We set $v$ to be the smallest in the ordering $\dot{<}$ and then recursively construct the remainder of the ordering on $G' = G \ \{v\}$.
\end{proof}
\begin{lemma}
	Let $G = (V, E)$ be $c$-uniformly sparse.
	After preprocessing in linear time and space, we can answer queries about the boundaries of subsets of vertices of $G$ in $O(ck)$ time, where $k$ is the size of the queried subset.
\end{lemma}
\begin{proof}
	Using Lemma \ref{lem_treelike} we first compute a $2c$-treelike ordering $\dot{<}$ of $V$.
	For every vertex $u \in V$, we store a list of edges $(u, v) \in E$ such that $u~\dot{<}~v$.
	We also compute and store the weighted degree $wd(v)$ for every vertex of $G$.
	
	Assume we are given $S \subseteq V$.
	We first compute the total weight $s_{internal}$ of edges internal to $S$.
	To this end, for every vertex $u \in S$ we go through its neighbors that are larger in the ordering $\dot{<}$ and sum up the weights of edges that lead to $S$.
	Note that every edge in $S \times S$ will be encountered exactly once.
	The boundary of $S$ can be computed as
	\begin{align*}
		s_{cut} := \left(\sum_{v \in S} wd(v)\right) - 2s_{internal}. 
	\end{align*}
\end{proof}
\begin{corollary}
	Given a graph with $n$ vertices, there exists a data structure that:
	\begin{itemize}
	  \item achieves the construction time, storage space, and cut approximation guarantees of Spielman-Srivastava sparsifiers, and
	  \item can compute approximate weights of cuts in $O(k \frac{\log n}{\epsilon^2})$ time, where $k$ is the size of the smaller side of the cut.
	\end{itemize}
\end{corollary}

\section{Final Note}
Similar techniques to those presented can show that if the vertices of graph $G$ are ordered by the sum of $R_ec_e$ on edges that have an endpoint of that vertex, any graph has a sparsifier that can be written as the union of $O(1/\epsilon^2)$ trees with that topological ordering on their vertices. The proof uses the same machinery as that presented above (with a slightly different use of the Marcus-Spielman-Srivastava sampling scheme), and we omit the full proof here.
\bibliographystyle{alpha}

\end{document}